\documentclass[<options>]{elsarticle}

\usepackage{amsfonts}
\usepackage{multirow}
\usepackage{eurosym}
\usepackage{hyperref}
\usepackage{graphicx, graphics}
\newtheorem{theorem}{Theorem}

\newtheorem{definition}[theorem]{Definition}

\newtheorem{remark}[theorem]{Remark}

\newenvironment{proof}[1][Proof]{\noindent\textbf{#1.} }{\ \rule{0.5em}{0.5em}}

\begin{document}

\begin{frontmatter}



\title{Entropy of spatial network with applications to non-extensive statistical mechanics}


\author[a]{O.K. Kazemi}
\ead{o.kazemi@stu.umz.ac.ir}
\affiliation[a]{organization={University of Mazandaran},
city={Babolsar},
country={Iran}}

\author[b]{ S.M. Taheri\corref{corb}}
\ead{sm_taheri@ut.ac.ir}
\affiliation[b]{organization={School of Engineering Sciences, College of Engineering},
addressline={University of Tehran},
postcode={1417613131},
city={Tehran},
country={Iran}}

\cortext[corb]{Corresponding author}

\begin{abstract}
A new method is proposed for analyzing complexity and studying the information in random geometric networks using Tsallis entropy tool. Tsallis entropy of the ensemble of random geometric networks is calculated based on the components of the random connection model on the point process which is obtained by connecting the points with a probability that depends on their relative positions (10.1016/j.indag.2022.05.002, 2022). According to information theory and conditional discussion, the bounds for Shannon and Tsallis entropies of the ensemble of this random graph are presented. Using this function and Lagrange's formula, the connection function that provides the maximum Tsallis entropy based on general constraints is obtained. Then, a simulation-based example is presented to clarify the application of the proposed method in the study of ad hoc wireless networks. By observing the obtained results, it can be stated that the wireless networks that adhere to the model studied here are almost maximally complex. Also, Tsallis conditional entropy maximizing function is compared with other connection functions using numerical calculations and the optimal value for the maximization of conditional entropies is obtained.
\end{abstract}



\end{frontmatter}


\section{ Introduction and bibliography}
The problem of graph entropy has a special importance and application in investigating the structure of networked systems. Consider the points placed in space according to a point process. For any two points $x$ and $y$, we insert an edge between these points with probability $g(|x-y|)$, independently of all other pairs of points. Here, $|x-y|$ denotes the Euclidean distance between $x$ and $y$. The connection function  $g: \mathbb{R}^d\rightarrow [0,1]$, $d \geq 2$ is a measurable symmetric function, also called the activation function \cite{AKB}. These edge connections lead to the formation of clusters of points, which is called as a soft random geometric graph (SRGG) \cite{41}. The SRGG or random connection model \cite{MH} of continuum percolation is a generalization of random geometric graph
 with two random sources, the point locations and their connections. The probability of existence of an edge between two points decreases as the distance between them increases \cite{94}. This model is quite general and has applications in different branches of science. In epidemiology, the probability that an infected herd at location $x$ infects another herd at location $y$ ; in telecommunications, the probability that two transmitters are non-shaded and can exchange messages ; in biology, the probability that two cells can sense each other \cite{87,60}. Also, this and related models have been studied in the contexts of geometric probability, statistics and physics \cite{80}. Examples of SRGG current application include modelling the collective behavior of multi-robot swarms, 
 brain connectivity and neural networks,  social networks, disease surveillance, electrical smart grid engineering and communication theory, where random geometric graphs have recently been used to model wireless networks, sharing information over communication channels which have a complex, stochastic impulse response \cite{New, 93, 92}.

In mathematical or physical discussions, there are several connection functions. The special case of the connection function is of the Boolean, zero$-$one type. Model made of this connection function, also known as the unit disk or Gilbert graph or Boolean model, was studied for the first time by Gilbert \cite{01} as a model of communication networks \cite{MH}. Such the special model of the SRGG is the simplest Boolean model of continuum percolation in stochastic geometry and percolation theory \cite{CSKM}, which is also known as the most basic random geometric graph and a central object in random graph theory. Another common choice of the connection function is the so-called Reyleigh fading connection function.
 This form of connection function is one of the most widely used connection functions in wireless communication networks \cite{WDG}, in which spatial network models are used to study the connection in these networks.

On the other hand, one of the tools for investigating the complexity and inherent information of systems, including connectivity in spatial networks, is the entropy graph. The entropy of a graph is an information theoretic functional which is defined on a graph with a probability density on its vertex set. This functional was originally proposed by J. K\"orner in 1973 to study the minimum number of codewords required for representing an information source \cite{Ko}. Evaluation of the distribution of rainfall stations,  wireless sensors networks, fractal traffic flows, evaluation of  network resilience, chemical networks by topological Indices, social networks, software risk analysis and communication network topology inference are among the applications of using entropy in the study of network systems \cite{ CD,SBB}. This diverse range of applications has led to the definition of several measures of entropy that allow one to quantify the complexity of a graph. Among these criteria, Shannon entropy \cite{SHa}, Renyi entropy and generalized entropy of Tsallis in non-extensive statistical mechanics \cite{TSa} can be mentioned. Tsallis or nonextensive entropy is one of the generalizations of Boltzmann-Gibbs entropy \cite{100} that has been interested in the last decades. Tsallis entropy has many applications in studying systems such as interconnected systems, cold atoms, plasmas, black holes, earthquakes , economics,  and so on \cite{104,CS}. 

Tsallis entropy of spatial network ensembles is definitely a new idea compared to the works done in describing the entropy of spatial networks and  the entropy of spatial network ensembles, which we discuss in this paper. Here, we develop the paper reported in \cite{CDG} based on Tsallis's work. That is, Tsallis entropy of spatial network ensembles is determined for the soft random geometric graph. These purposes give us a motivation to organize this paper as follows. In sections \ref{sec2} and \ref{sec3}, we describe the entropy values and the basic model used in this paper, respectively. According to the information theory, we present an upper bound for Tsallis entropy of ensembles. Then, we evaluate Tsallis entropy from a conditional point of view. According to this formula, we obtain the connection function that maximizes Tsallis entropy under general constraints. These main results are discussed in Section \ref{sec4}.  In the following, we describe the bounds of the conditional entropy and Tsallis conditional entropy, which are clearly shown in Section \ref{sec5} by drawing a diagram. In Section  \ref{sec5}, we also apply our stated analytical methods to the study of ad hoc wireless networks.

\section{Entropy and generalized entropies}\label{sec2}
Entropy is a criterion for measuring random variable uncertainty first introduced by Shannon in 1948 \cite{SHa}. Over the past decades, some researchers have extended this concept for measuring uncertainty, among them Rényi \cite{YD} and Tsallis \cite{TSa}.
In this section, we briefly review some essential definitions and results regarding the concept of entropy.
\begin{definition}
If $X$ is a discrete random variable with the probability mass function $p_x=p_X(x)$, then Tsallis entropy of order $\alpha$ in the discrete system is defined as
\begin{eqnarray}\label{Ts}
 H^{T}(p,q)=\frac{1}{1-q}\Big(\sum_{x}{p_x^{q}-1\Big)}=-\sum_{x}{p_{x}^{q}\ln_q{p_{x}}  },
\end{eqnarray}
where for a real number $x\geq 0$ and $q\in R$,
\begin{eqnarray*}
\ln_q{x}= \left \{ \begin{array}{ll}
\ln{x}&q=1\\
\frac{x^{1-q}-1}{1-q}&o.w.
\end{array} 
\right.
\end{eqnarray*}
\end{definition}

The function $\ln_q{x}$ is the mathematical basis for Tsallis statistics and is called the $q$-logarithm, which satisfies the axioms for non-extensive entropies introduced in \cite{105}. It is worth noting that Tsallis entropy of order $q$ contains the Boltzmann-Gibbs statistics, which is also referred to as Shannon entropy when $q\rightarrow 1$.
\begin{definition}
Let $X$ be a discrete random variable drawn according to the probability density function $p_x=p_X(x)$. Shannon entropy of $X$, $H^S(p)$, is defined as the expected value of the random variable $\ln{1/p(x)}$, $\mathbb{E}(\ln{1/p(x)})$, so that we have
\begin{eqnarray}\label{eqS}
 H^S(p)=-\sum_{x}{p_x\ln p_x}.
\end{eqnarray} 
\end{definition}

The extensive Shannon entropy and standard definition of the expectation value are recovered in Relation \ref{eqS}, so $(1-q)$ in Relation \ref{Ts} can be interpreted as a measure of the lack of extensivity of the system. The parameter $q$ is used to adjust the measure depending on the shape of the probability distributions. Using Tsallis entropy, for $q > 1$, high-probability events contribute more to the entropy value than do low-probability events. Therefore, the higher the value of $q$, the greater the contribution of high probability events to the final result \cite{MAM}.
\subsection*{Erd\H os-R\'enyi model }
Let us begin with a brief explanation of our underlying situation. As an example, we consider the Erd\H os-R\'enyi  (ER) random graph model, which was the first model presented in 1959 for modeling complex networks. In this model, each edge exists independently with a probability $p$. In the part $(a)$ of Figure \ref{ERS}, ER is shown with $p= 0.01$ and $200$ nodes. 
If we have $n$ nodes, the total number of edges is equal to $c_n=n(n-1)/2$, so that different combinations of these edges form the graphs in the ensemble $\mathcal{G}$. Therefore, Shannon entropy and Tsallis entropy of the ER ensemble are obtained respectively as
\begin{eqnarray}\label{PD}
H^S(\mathcal{G})=-c_n \big( p \ln p +(1-p) \ln (1-p)\big)
\end{eqnarray}
and,
\begin{eqnarray}\label{PD}
H^T(\mathcal{G})=\frac{c_n}{1-q}\big[p^q+(1-p)^q-1\big].
\end{eqnarray}

\section{The Based Model of  Point Process}\label{sec3}
Consider a set $\Phi$ of $n$ nodes placed in a space $\mathcal{K} \subset \mathbb{R}^d$ with a volume $K = vol(\mathcal{K})$.
The locations of the nodes $x_1,\dots,x_n$ form a point process in $\mathcal{K}$.
\begin{figure}[h]
\begin{center}
\includegraphics[width=\linewidth]{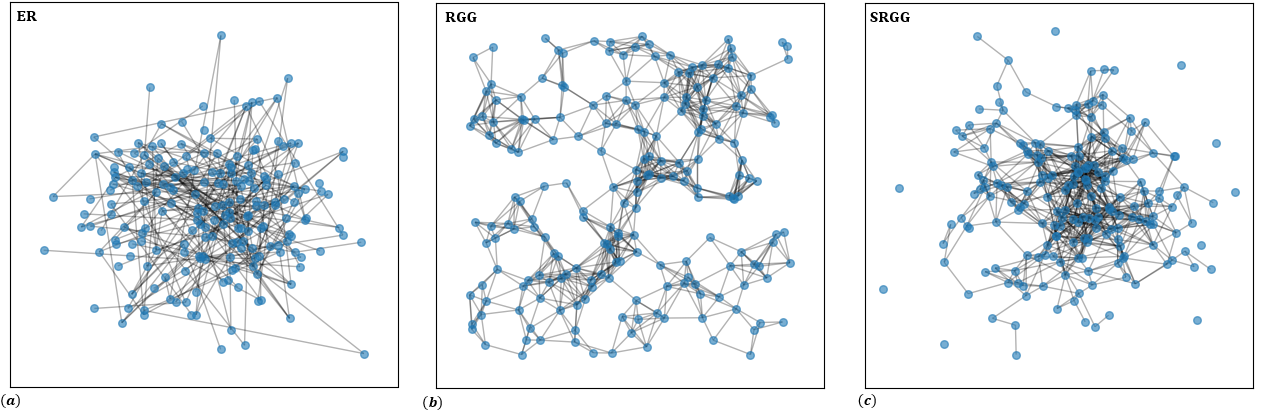}
\end{center}
\caption{$a)$ Random graph, $b)$ Random geometric graph and $c)$ Soft random geometric graph with $n=200$ nodes. }
\label{ERS}
\end{figure}
In this paper, we consider a simple, uniform point process which demonstrates that the spatial distribution of nodes is described by a constant and finite intensity of $\rho=n/K$ nodes per unit volume. Hence, the triple $(\Phi,\mathcal{K},\rho)$ describes a binomial point process. Given two points  $x_i$ and $x_j$  of $\Phi$ (with $i >j$), connect them by an edge with probability $g({x_{i,j}})=g(|x_i-x_j|)$  independently of all other pairs of points. The symbol $|.|$ denotes the Euclidean norm. This yields an undirected random graph $G(\Phi )\in \mathcal{G}$ with a vertex set $\Phi$, where $\mathcal{G}$ denotes the set of all graphs. Actually, Gilbert's graph or the random geometric graph is the special case with the connection function $g(x)= I_{\{|x|\leq x_0\}}$, where $x_0$ denotes the maximum connection range \cite{01}, and it is shown in the part $(b)$ of Figure \ref{ERS}. For other connection functions $g$, e.g., monotonically decreasing functions in the distance argument, we have the soft random geometric graph \cite{41}, as shown  in the part $(c)$ of this figure.
\section{Connection probability by Tsallis entropy maximization}\label{sec4}
As we mentioned earlier, we are interested in quantifying the Tsallis entropy of the ensemble $\mathcal{G}$. Consider $P(G)$ as the probability of occurrence of graph $G$ for each graph $G\in \mathcal{G}$. It is clear that this probability depends on the two items of the spatial distribution of the nodes and the pair connection function $g$. The definition of Tsallis entropy of $\mathcal{G}$ is defined as
\begin{eqnarray}
 H^{T}(\mathcal{G})=\frac{1}{1-q}\Big(\sum_{G\in \mathcal{G}}{P(G)^{q}-1\Big)}.
\end{eqnarray}
It is noteworthy that the existence of the edge $(i,j)$ is a function of the distance pair $x_{ i,j}$. Without losing generality, we consider the graph as a random variable G with support $\mathcal{G}$ whose distribution is equivalent to the distribution of the edge set. In this case, let $Y_{ i,j}$ represent a Bernoulli random variable that models the existence of an edge $(i,j)$ which is a function of the pair distance $x_{ i,j}$. We can write Tsallis entropy of the network ensemble as the joint Tsallis entropy of the sequence $\{ Y_{ i,j} \}$, i.e.,
\begin{eqnarray}\label{e1}
H^T(\mathcal{G})=H^T(Y_{1,2},Y_{1,3},\dots,Y_{n-1,n})
\end{eqnarray}
The random variable $Y_{ i,j}$ is physically related to nodes $i$ and $j$, but it should be expressed that $P(Y_{ i,j}= 1)\neq g(x_{ i,j} )$, since this probability depends on the pair distances $x_{ i,j}$ and the node locations \cite{CDG}. More precisely, the probability of $Y_{ i,j}$  is averaged over all pair distances, which is obtained as
\begin{eqnarray}\label{PDP}
P(Y_{ i,j}= 1)=\int_{0}^D f(x_{i,j})g(x_{i,j})dx_{i,j},
\end{eqnarray}
where, $D=sup_{{\bold{x}_i,\bold{x}_j}\in \mathcal{K}}|\bold{x}_i-\bold{x}_j|$ and $ f(x_{i,j})$ is the pair distance probability density corresponding to nodes $i$ and $j$. Considering that the functions $f(x_{i,j})$ and $g(x_{i,j})$ are identical for all $i \neq j$, we can call Relation \ref{PDP} $\bar{p}$ \cite{CDG}.
\subsection*{Tsallis conditional entropy}
To study the effect of a particular embedding on Tsallis entropy of the SRGG ensemble, we consider the information-theoretic concept of Tsallis conditional entropy. Here, given the distribution of vertex locations, Tsallis conditional entropy of the graph ensemble is defined as
\begin{eqnarray}\label{8}
H^T(\mathcal{G}|\mathcal{X})&=& \left \langle H^T(\mathcal{G}|\bold{x}_1,\dots,\bold{x}_n)\right \rangle \\ \nonumber
&=& \left \langle H^T(\mathcal{G}|x_{1,2},x_{1,3},\dots,x_{n-1,n})\right \rangle,
\end{eqnarray}
where the notation $\left \langle.\right \rangle$ denotes the average of a function over the vertex positions. Also, this quantity is written in terms of an average over the pair distances $x_{i,j}$ with the joint pair distance density $f(x_{1,2},x_{1,3},\dots,x_{n-1,n})$. \\
Now, we state and prove a main result concerning the lower and upper bounds of Tsallis entropy of the graph ensemble.
\begin{theorem}\label{THE3}
Consider the graph G with support $\mathcal{G}$ and Bernoulli random variables of $Y_{ i,j}$. Then, the following relation is valid for Tsallis entropy
\begin{eqnarray}\label{9}
H^T(\mathcal{G}|\mathcal{X})\leq H^T(\mathcal{G})\leq c_n H^T_2(\bar{p}).
\end{eqnarray}
\end{theorem}
\begin{proof}
For the right inequality in Relation \ref{9}, using the extended relation, often called independence bound on entropy in information theory and Relation \ref{e1}, we have for $q \geq 1$,
\begin{eqnarray}\label{e11}
H^T(\mathcal{G})&\leq& \sum_{i<j}H^T(Y_{i,j})\\ \nonumber
&\leq& c_n H^T_2(\bar{p}),
\end{eqnarray}
where the binary Tsallis entropy function, $H^T_2(\bar{p})$, is
\begin{eqnarray}
H_2^T(\bar{p})=\frac{1}{1-q}\big[\bar{p}^q+(1-\bar{p})^q-1\big].
\end{eqnarray}
Note that if $q=1$ and the random variables are independent, the equality holds. It is clear that the bound is maximized when $\bar{p} = 1/2$.

 For the left-hand  inequality in Relation \ref{9}, a fundamental result of information theory states that conditioning reduces uncertainty \cite{CT}. This expression is valid for Tsallis entropy for $q \geq 1$, which means that the conditioning reduces Tsallis entropy \cite{F}. Therefore, we have
\begin{eqnarray}\label{e110}
H^T(\mathcal{G}|\mathcal{X})\leq H^T(\mathcal{G}).
\end{eqnarray}
As a result, according to Relations \ref{e11} and \ref{e110}, the theorem is proved.
\end{proof}
\begin{remark}\label{Rem4}
Using Theorem $3.7$ in \cite{F} and for $q \geq 1$, we can write
\begin{eqnarray}\label{e2}
H^T(\mathcal{G}|x_{1,2},\dots,x_{n-1,n}) &=& H^T(Y_{1,2}|x_{1,2},\dots,x_{n-1,n}) \nonumber \\ 
&\qquad+& H^T(Y_{1,3}|Y_{1,2}, x_{1,2},\dots,x_{n-1,n}) +\dots \nonumber \\ 
&\qquad+&H^T(Y_{n,n-1}|Y_{n-1,n-2},\dots, x_{1,2},\dots,x_{n-1,n}) \nonumber \\ 
&\qquad \leq& \sum_{i<j} H^T(Y_{i,j}|x_{i,j}),
\end{eqnarray}
The edge states $\{Y_{i,j}\}$ are independent conditioned on the pair distances $\{x_{i,j}\}$. Thus, the equality holds when $q\rightarrow 1 $. According to Relation \ref{8} and averaging the right-hand side of Relation \ref{e2} over the density $f(x_{1,2},\dots,x_{n-1,n})$ leads to
 \begin{eqnarray}\label{e3}
H^T(\mathcal{G}|\mathcal{X})\leq c_n\int_{0}^{D}f(x)H^T_2\big(g(x)\big)dx.
\end{eqnarray}
\end{remark}

It is worth noting that Shannon conditional entropy is given as \cite{CDG}
 \begin{eqnarray}\label{e33}
H^S(\mathcal{G}|\mathcal{X})=c_n\int_{0}^{D}f(x)H_2^S\big(g(x)\big)dx; H_2^S\big(g\big)=g\log(g)-(1-g)\log(1-g).
\end{eqnarray}
\begin{remark}
In the inequality \ref{9}, if the connection function $g(x)=1$ for $0\leq x \leq x_0$ and $g(x)=0$ otherwise, then the lower bound is $H^T(\mathcal{G}|\mathcal{X})=0$. It can be easily seen that Tsallis entropy maximizing function is $g(x)=\frac{1}{2}$. In this case and when $q\rightarrow 1 $
\begin{eqnarray}
H^T(\mathcal{G}|\mathcal{X}) = H^T(\mathcal{G}) = \frac{c_n}{1-q}[{ ( \frac{1}{2}) }^{q-1}-1 ].
\end{eqnarray}
Also, we have $H^S(\mathcal{G}|\mathcal{X}) = H^S(\mathcal{G}) = c_n \ln 2$ when $q$ tends to 1.
\end{remark} 

However, practical spatial networks usually exhibit certain properties, which we may wish to incorporate into the maximization task through some constraints. For example, one of these constraints can be the mean degree of a geometric graph ensemble. In general, we consider a graph ensemble where the pair connection function obeys a set of constraints
\begin{eqnarray}\label{e4}
\int_{0}^{D}\theta_\ell(x)g(x)^q dx=k_{\ell,q},  \ell=1,2,\dots,L
\end{eqnarray}
where $\{\theta_\ell \}$ and $\{k_\ell \}$ are independent of $n$. Now, we find the maximizing function $g(x)$ given these constraints. 
\begin{remark}
Euler-Lagrange equation according to the constraints given in Relation \ref{e3}, when q tends to 1 and Relation \ref{e4} is defined as
\begin{eqnarray}
H^T(\mathcal{G}|\mathcal{X})-\sum_{\ell=1}^{L}\Big[\lambda_\ell \big[\int_{0}^{D}\theta_\ell(x)g(x)^qdx-k_{\ell,q}\big]\Big]=0
\end{eqnarray}
Solving this equation yields the maximum Tsallis entropy connection probability
\begin{eqnarray}\label{mT}
 \tilde{g}(x)=\frac{1}{\Big(1-(1-q)\psi(x)\Big)^{\frac{1}{q-1}}+1},
\end{eqnarray}
where
\begin{eqnarray}
 \psi(x)=\frac{1}{c_nf(x)} \sum_{\ell=1}^{L}\lambda_\ell\theta_\ell(x),
\end{eqnarray}
and $\{\lambda_\ell \}$ are uncertain Lagrange multipliers that depend on the geometric properties of the network. The maximizing pair connection function $ \tilde{g}$ is dependent on the pair distance density $f$ and the constraints captured by $\{\theta_\ell \}$. Of course, by choosing $\{\theta_\ell \}$ as a multiple of the function $f$, we no longer have this dependence.
\end{remark} 

Another result deals with the description of bounds in Shannon conditional entropy and Tsallis conditional entropy. Notably, a bound of Shannon conditional entropy is given in \cite{CBG}.
\begin{theorem}\label{TH5}
In the constraints given in Relation \ref{e4}, if $\theta_\ell(x)=f_X(x)$ for $\ell=1,\dots,L$  be the pair distance probability distribution function for two points in $\Phi$ then the $q$th moment of the pair connection function is defined as
\begin{eqnarray}
\mu_q=\int_0^{D}f_X(x)g(x)^qdx, \qquad          q>0
\end{eqnarray}
then,
\begin{eqnarray}
 H^S(\mathcal{G}|\mathcal{X}) &\geq& n(n-1)(\mu_1-\mu_2)  \\ 
  H^T(\mathcal{G}|\mathcal{X})&\leq& c_n.
\end{eqnarray}
\end{theorem}
\begin{proof}
To prove the first inequality, it can be shown that
\begin{eqnarray}\label{17}
H_2^S(x)>2x(1-x) ,        \qquad       0<x<1.
\end{eqnarray}
Hence, we have
\begin{eqnarray*}
 H^S(\mathcal{G}|\mathcal{X})=c_n\mathbb{E}[H_2^S(g(x))]&=&c_n\int_0^D f_X(x)H_2^S(g(x))dx\\ \nonumber
&\geq&2c_n\int_0^D f_X(x)g(x)(1-g(x))dx\\ \nonumber
&=&n(n-1)(\mu_1-\mu_2)
\end{eqnarray*}
On the other hand, the Tsallis entropy $H_2^T(x)$ satisfies
\begin{eqnarray}\label{19}
H_2^T(x)\leq 1,        \qquad       0<x<1.
\end{eqnarray}
Therefore, according to relation \ref{19}, 
 \begin{eqnarray}\label{e3}
 c_n\int_{0}^{D}f(x)H^T_2\big(g(x)\big)dx \leq c_n\int_0^D f(x) dx,
\end{eqnarray}
where the second inequality of the theorem is obtained.
\end{proof}

\section{Simulation and studying a practical example}\label{sec5}
In this section, simulation studies are presented related to the results obtained in the previous section. Programs related to this paper are provided using R 4.3.0 and Python software. We applied our analytical methods to study the example of real-world spatial networks, communication networks.
\subsection{Simulation of conditional entropies of Tsallis and Shannon}
In mathematical or physical discussions, there are several connection functions \cite{92}. One of the common choices of connection function is the Rayleigh fading connection function \cite{AKB}, i.e.,
\begin{eqnarray}\label{CF}
g(x)=\beta e^{-{(\frac{x}{x_0})}^\eta}.
\end{eqnarray}
For this function, the path loss power is measured by $\eta>0$. $\beta\in (0,1]$ specifies the edge probability for short-range connections. In this paper, we consider $\beta=1$. This form of connection function is one of the most widely used connection functions in the wireless communication literature \cite{WDG}. It is known that the farther the distance between the nodes, the lower the probability of connecting two points. Note that by letting $\eta \rightarrow \infty$, we recover the hard disk connection function with radius $x_0$. This parameter mathematically controls the stretch of the decaying exponential. In order to calculate the effect that this parameter has on Tsallis conditional entropy in each pairwise connection
\begin{figure}[h]
\begin{center}
\includegraphics[width=\linewidth]{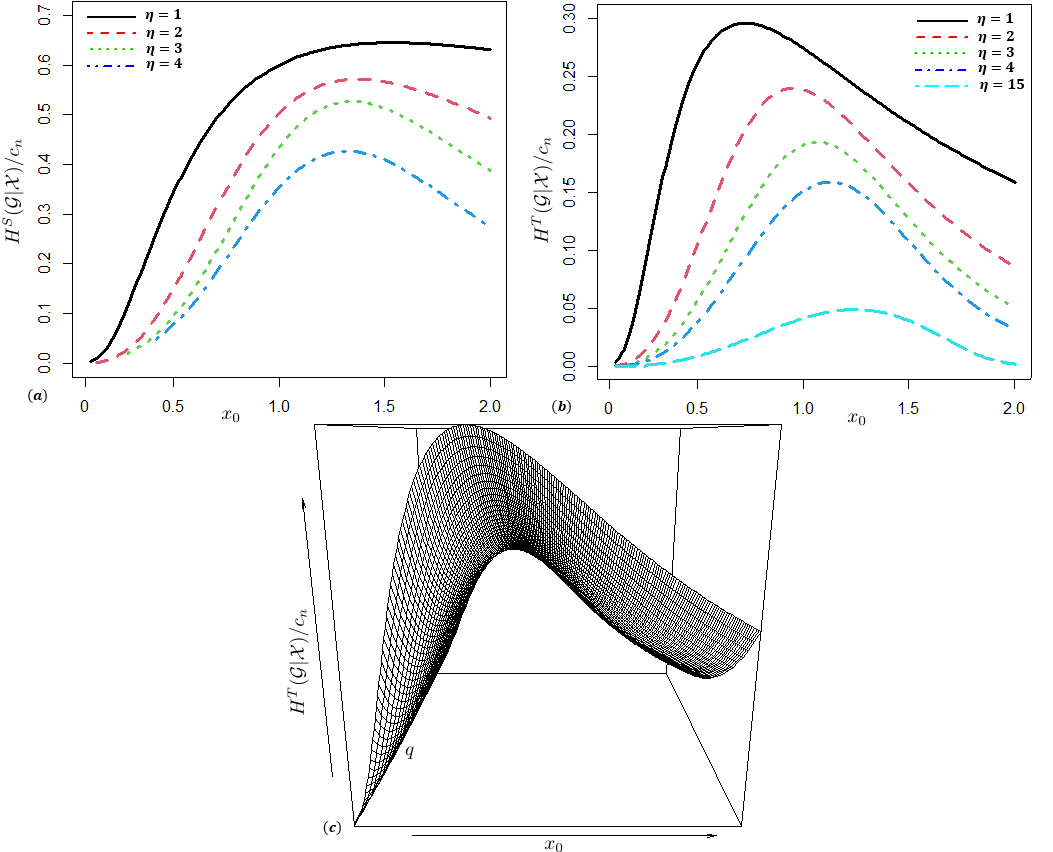}
\end{center}
\caption{$a)$ Shannon conditional  entropy, $H^S(\mathcal{G}|\mathcal{X})/c_n$, $b)$ Tsallis conditional entropy, $H^T(\mathcal{G}|\mathcal{X})/c_n$ and $c)$ the upper bound of Tsallis conditional entropy with respect to $x_0\in (0,2]$ and $q\in (1,4]$.}
\label{fig1}
\end{figure}
, we consider a network located in a sphere with a unit radius in three dimensions. In this case, the pair distance distribution is \cite{CDG, TF}
\begin{eqnarray}
f(x)=3x^2(1-\frac{3}{4}x+\frac{1}{16}x^3).
\end{eqnarray}

Tsallis conditional entropy $H^T(\mathcal{G}|\mathcal{X})$ in Inequality \ref{e3} divided by the total number of possible edges against the typical connection range for a wireless communication network is plotted in the part $(b)$ of Figure \ref{fig1}.  Here, $q$ is chosen very close to 1 so that the equality in relation \ref{e3} is established. This function is plotted for different values of $\eta$. The connection probability $g$ increases with an increase in $\eta$. As a result, it is expected that Tsallis conditional entropy will have a non-increasing trend with increasing $\eta$, which is clearly shown in the part $(b)$ of this figure. In the part $(a)$ of this figure, Shannon conditional entropy in Relation \ref{e33} is drawn for different values $\eta$. In addition to obtaining the results stated in the paper \cite{CDG}, we have performed more recent work based on Tsallis entropy. The remarkable thing about the use of Tsallis entropy is that the observed changes of $H^T(\mathcal{G}|\mathcal{X})/c_n$ in terms of the parameter $\eta$ are similar to the changes observed in $H^S(\mathcal{G}|\mathcal{X})/c_n$, with the difference that this function can be calculated and plotted for values of $\eta>4$ in the range of $x_0\in (0,2]$. 

To investigate how the upper bound of Tsallis conditional entropy changes according to the quantities $x_0\in (0,2]$ and $q\in (1,4]$, the upper bound of the function $H^T(\mathcal{G}|\mathcal{X})$ as a bivariate function of these quantities was calculated and drawn in the part $(c)$ of Figure \ref{fig1}. As can be seen, Tsallis conditional entropy decreases as the values of $q$ and $x_0$ increase.

\subsection{The connection function based on Tsallis entropy and other connection functions}\label{subsec2}
In this subsection, we compare two functions: the pair connection probability \ref{CF} and the maximum Tsallis entropy connection probability \ref{mT} obtained from Lagrange equation. This comparison is shown in Figure \ref{fig3}. 
\begin{figure}[h]
\begin{center}
\includegraphics[width=\linewidth]{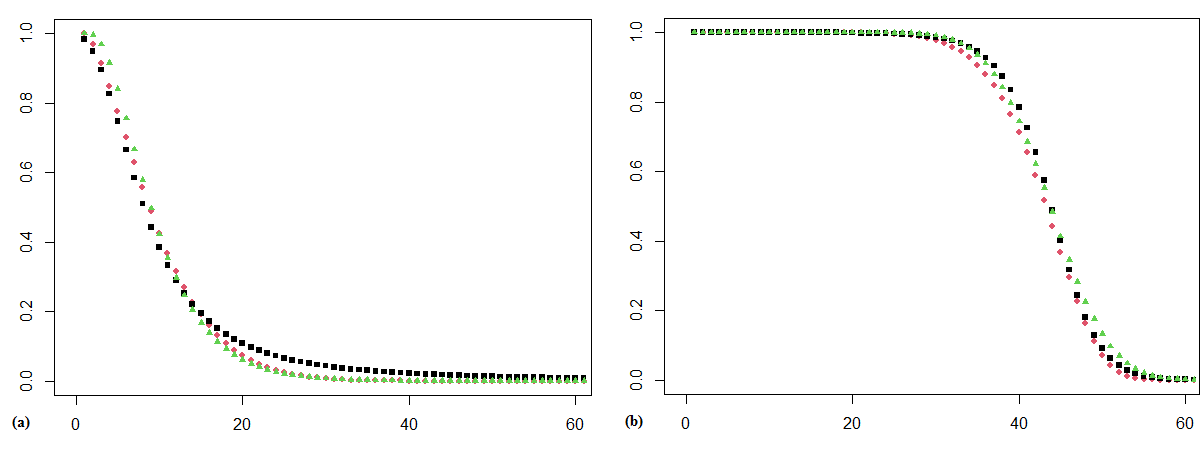}
\end{center}
\caption{The pair connection probability defined in Relation \ref{CF} (diamond),
the maximum Tsallis entropy connection probability (square) and MIMO function (triangle)}
\label{fig3}
\end{figure}
In the part $(a)$, we let $\eta=1.5$, $x_0=0.1$ and $x\in (0,1)$  for Reyleigh connection function with a diamond symbol. The maximum Tsallis entropy connection probability $\tilde{g}$ is plotted for $q=1.4$ with a square symbol. As can be seen, the values of these functions are very close to each other so that the mean squared error of these functions is equal to $0.00049$. The results of these comparisons based on criterion $MSE(p_1(x),p_2(x))=\sum_x(p_1(x)-p_2(x))^2/n$ are shown in Table \ref{tab1}. In the part $(b)$, these functions are plotted for the values of  $\eta=9$, $x_0=2.2$ and $x\in (0,4.5)$ and for the value $q$ very close to 1.  In this case, it is clear that the Reyleigh connection function is extremely close to Tsallis entropy maximizing function. 
\begin{center}
\begin{table}[h!]%
\centering
\caption{ Mean Squared Error.\label{tab1}}%
\begin{tabular*}{280pt}{@{\extracolsep\fill}lcc@{\extracolsep\fill}}
\hline\hline
 &\textbf{$\tilde{g}$} & MIMO    \\
\hline\hline  
 & $(\eta=1.5,x_0=0.1,q=1.4)$  &$(\eta=0.84,x_0=0.1)$    \\
 & 0.00049  & 0.00023    \\
Reyleigh& $(\eta=9,x_0=2.2,q\rightarrow 1)$  & $(\eta=4,x_0=1.6)$    \\
 & 0.00064  & 0.00044   \\
\hline
 &  & $(\eta=0.84,x_0=0.1)$    \\
 &   & 0.00098     \\
 $\quad\tilde{g}$&   & $(\eta=4,x_0=1.6)$     \\
 &   & 0.00021     \\
\end{tabular*}
\end{table}
\end{center}
In fact, the mean squared error in this case is $0.00064$. Therefore, it can be stated that wireless networks adhering to the model studied here are almost maximally complex.

In addition, using experimental and practical observations, we found a close relationship between the function $\tilde{g}$ and other connection functions. Specifically, we compared the maximum Tsallis entropy connection probability $\tilde{g}$ with another connection function under title the MIMO (multiple input and multiple output) connection function \cite{92}, that is defined as
\begin{eqnarray}\label{MI}
g(x)=e^{-{(\frac{x}{x_0})}^\eta} \Big( 2+(\frac{x}{x_0})^{2\eta}-  e^{-{(\frac{x}{x_0})}^\eta}  \Big).
\end{eqnarray}
This function is shown in Figure \ref{fig3} with a triangle symbol. Comparing the MIMO connection function with the maximum Tsallis entropy connection probability \ref{mT}, the mean square error is obtained in the parts $(a)$ and $(b)$ of Figure \ref{fig3}, respectively, $0.00098$ with $x_0=0.1,\eta=0.84$ and $0.00021$ with $x_0=1.6,\eta=4$. The results of these comparisons are mentioned in Table \ref{tab1}.
\subsection{The optimal value of the maximization of conditional entropies}\label{subsec3}
Continuing the results obtained from simulations and numerical calculations, we seek the optimal value of $x_0$ to maximize the conditional entropies. We calculated the conditional entropy $H^S(\mathcal{G}|\mathcal{X})$ for different $\eta \in [0.5,3.9]$, with $x_0$ belonging to the interval $(0,2]$. We derived the $x_0$ that maximizes the conditional entropy for each $\eta$ and plotted it in Figure \ref{fig4}.
\begin{figure}[h]
\begin{center}
\includegraphics[width=\linewidth]{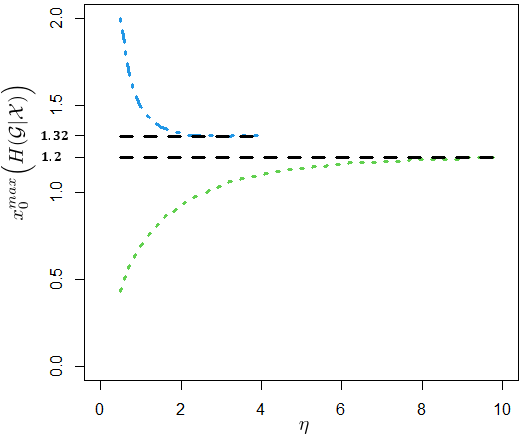}
\end{center}
\caption{The values $x_0$ of the function maximizer $H^S(\mathcal{G}|\mathcal{X})$ (dotdash) and the function maximizer  $H^T(\mathcal{G}|\mathcal{X})$ (dotted) with $q=1.1$ in terms of $\eta$.}
\label{fig4}
\end{figure}
 As it is clear, with an increase in $\eta$, the value $x_0$ of the function maximizer $H^S(\mathcal{G}|\mathcal{X})$ $\Big(=x_0^{max}(H^.(\mathcal{G}|\mathcal{X})\Big)$ decreases and has a non-increasing trend (dotdash). These calculations are also performed for Tsallis conditional entropy and for $x_0$ belonging to interval $(0,2]$, which is shown in Figure \ref{fig4} with the dotted line. In this case, unlike the previous one, with an increase in $\eta$, the value $x_0$ of the function maximizer $H^T(\mathcal{G}|\mathcal{X})$ increases, so we observe a non-decreasing trend. It can be concluded that, in this discussed network, the value $x_0$ for which the functions of Shannon conditional entropy and Tsallis conditional entropy in terms of the different values $\eta$ take their maximum value in the critical interval $(1.2,1.3)$.
\subsection{Simulation of  the conditional entropies bounds}
In Section \ref{sec4} of this article, the bounds of Shannon conditional entropy and Tsallis conditional entropy are stated in the form of Theorem \ref{TH5}, and these results are given in Relations \ref{17} and \ref{19}, respectively. The function $H_2^S(x)/2x(1-x)$ for $x \in (0,1)$ is plotted in Figure \ref{fig5} with the dashed line and above the solid line $y=1$. 
\begin{figure}[h]
\begin{center}
\includegraphics[width=\linewidth]{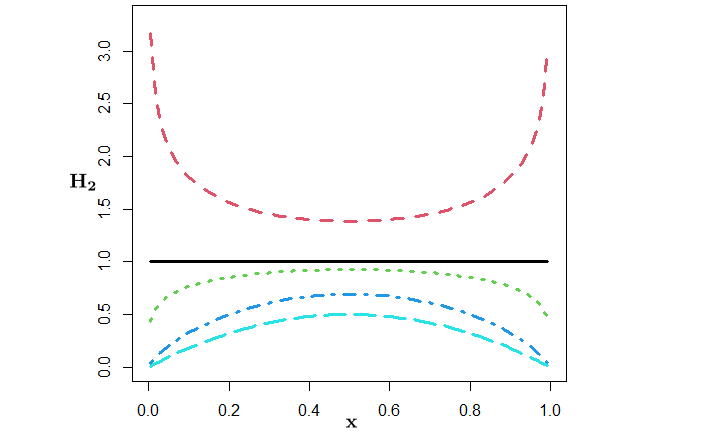}
\end{center}
\caption{$H_2(x)/2x(1-x)$ for $x \in (0,1)$(dashed line) and  $H_2^T(x)$  for $q=0.2$ (dotted line), $q= 0.999$ (dotdash line) and $q=2$ (longdash line).}
\label{fig5}
\end{figure}
Also, the function $H_2^T(x)$ is calculated for three different values of $q=0.2, 0.999, 2$ and shown with dotted, dotdash and longdash lines, respectively. As can be seen, the values of Tsallis entropy for each value $q$ are placed under the line $y=1$. Based on these numerical results, the results stated in Theorem \ref{TH5} are clearly investigated.

\section{ Conclusion}
We propose a method based on Tsallis entropy for analyzing the complexity of spatial network ensembles. We study the components of the soft random geometric graph on point process, which is obtained by connecting the points with a probability that depends on their relative position. We use the concepts of Tsallis entropy and Tsallis conditional entropy in this model. The obtained results can be summarized as follows:
\renewcommand{\labelenumi}{\Roman{enumi})}
\begin{enumerate}
\item  We stated and proved a main result concerning the lower and upper bounds of Tsallis entropy of the graph ensemble in Theorem \ref{THE3}.
\item Another result deals with the description of the bound in Shannon conditional entropy and Tsallis conditional entropy in Remark \ref{Rem4} and Theorem \ref{TH5}.
\item  The connection function that provides the maximum Tsallis conditional entropy under general constraints is calculated.  
\item These analytical methods are used to study ad hoc wireless networks. We compare the function obtained from the Lagrange equation with the pair connection probability and the MIMO connection function in Subsection \ref{subsec2}.  Experimental evaluations show the ability of the results in estimating the maximum complexity of wireless networks. 
\item As another outcome, using numerical calculations, we seek the optimal value $x_0$ to maximize the conditional entropies  in Subsection \ref{subsec3}.
\end{enumerate}







%

\end{document}